\documentclass[reqno]{amsart}

\usepackage{booktabs}
\usepackage{IEEEtrantools}
\usepackage{amssymb,latexsym,amsfonts,amsmath}
\usepackage{mathrsfs}
\usepackage{graphicx}
\usepackage{fancyhdr}%%
\usepackage{subfig}%%
\usepackage{dsfont}
\usepackage{tikz}
\usetikzlibrary{calc,shapes,arrows}

\definecolor{mediumblue}{rgb}{0.0, 0.0, 0.8}
\definecolor{mediumcandyapplered}{rgb}{0.89, 0.02, 0.17}
\definecolor{nazar}{rgb}{0.7, 0.5, 0.9}
\makeatletter
\let\NAT@parse\undefined
\makeatother
\usepackage[colorlinks=true, citecolor=black, linkcolor=mediumblue, urlcolor = nazar, final]{hyperref}

\definecolor{lightblue}{rgb}{0.30,0.75,0.93}

\DeclareRobustCommand\sampleline[1]{%
	\tikz\draw[#1] (0,0) (0,\the\dimexpr\fontdimen22\textfont2\relax)
	-- (2em,\the\dimexpr\fontdimen22\textfont2\relax);%
}

\usepackage{algorithm}
\usepackage{algorithmic}

\usepackage{xspace}

\newtheorem{theorem}{Theorem}[section]
\newtheorem{lemma}[theorem]{Lemma}

\newtheorem{problem}[theorem]{Problem}

\newtheorem{definition}[theorem]{Definition}

\newtheorem{remark}[theorem]{Remark}

\numberwithin{equation}{section}

\newcommand{\R}{{\mathbb{R}}}

\newcommand{\myexpression}{%
	\left( \mathcal X_{1, T} Q - \bar{\mathcal{X}}_{1, T}\bar Q \right) \left( \mathcal   U_{0, T} Q \right)^{\dagger} \! \left(\Xi \hat x + \Psi \hat u\right) %
}

\newcommand{\Bdata}{%
	\left( \mathcal X_{1, T} Q \!-\! \bar{\mathcal{X}}_{1, T} \bar Q \right) \! \left( \mathcal   U_{0, T} Q \right)^{\dagger}%
}

\usepackage{fancyhdr}

\newenvironment{nouppercase}{%
	\renewcommand{\uppercasenonmath}[1]{}}{}

\usepackage{amsmath}
\usepackage[many]{tcolorbox}
\usetikzlibrary{calc}
\tcbuselibrary{skins}

\newtcolorbox{resp}[1][]{%
	enhanced jigsaw,%
	colback=gray!5!white,%
	colframe=gray!80!black,%
	size=small,%
	boxrule=1pt,%
	halign title=flush center,%
	coltitle=black,%
	breakable,%
	drop shadow=black!50!white,%
	attach boxed title to top left={xshift=1cm,yshift=-\tcboxedtitleheight/2,yshifttext=-\tcboxedtitleheight/2},%
	minipage boxed title=3cm,%
	boxed title style={%
		colback=white,%
		size=fbox,%
		boxrule=1pt,%
		boxsep=2pt,%
		underlay={%
			\coordinate (dotA) at ($(interior.west) + (-0.5pt,0)$);
			\coordinate (dotB) at ($(interior.east) + (0.5pt,0)$);
			\begin{scope}[gray!80!black]
				\fill (dotA) circle (2pt);
				\fill (dotB) circle (2pt);
			\end{scope}
		}%
	},%
	#1%
}

\DeclareRobustCommand{\legendsquare}[1]{%
	\textcolor{#1}{\rule{1.5ex}{1.5ex}}%
}%%

\DeclareRobustCommand{\legendcircle}[1]{%
	\tikz[baseline=-0.75ex]\node[circle, draw, fill=#1, #1, minimum size=0.15ex] at (0,0) {};%
	\nobreak\hspace{0.3em}%
}

\definecolor{START}{rgb}{0.4660 0.6740 0.1880}%%
\definecolor{TARGET}{rgb}{0 0.4470 0.7410}%%
\definecolor{OBSTACLES}{rgb}{0.93,0.69,0.13}%%

\definecolor{START1}{rgb}{0, 1, 0}%%
\definecolor{TARGET1}{rgb}{0.3010 0.7450 0.9330}%%
\definecolor{OBSTACLES1}{rgb}{1, 0, 0}%%

\begin{document}

\begin{abstract}
Model order reduction (MOR) involves offering \emph{low-dimensional models} that effectively approximate the behavior of complex high-order systems. Due to potential model complexities and computational costs, designing controllers for high-dimensional systems with complex behaviors can be challenging, rendering MOR a practical alternative to achieve results that closely resemble those of the original complex systems. To construct such effective reduced-order models (ROMs), existing literature generally necessitates precise knowledge of original systems, which is often unavailable in real-world scenarios. This paper introduces a data-driven scheme to construct ROMs of dynamical systems with \emph{unknown} mathematical models. Our methodology leverages data and establishes similarity relations between output trajectories of unknown systems and their data-driven ROMs via the notion of \emph{simulation functions (SFs)}, capable of formally quantifying their closeness. To achieve this, under a rank condition readily fulfillable using data, we collect only \emph{two input-state trajectories} from unknown systems to construct both ROMs and SFs, while offering \emph{correctness guarantees}. We demonstrate that the proposed ROMs derived from data can be leveraged for controller synthesis endeavors while effectively ensuring \emph{high-level logic} properties over unknown dynamical models. We showcase our data-driven findings across a range of benchmark scenarios involving various \emph{unknown} physical systems, demonstrating the enforcement of diverse complex properties.
\end{abstract}

\title{{\LARGE Model Order Reduction from Data with Certification\vspace{0.4cm}}}

\author{{\bf {\large Behrad Samari}}}
\author{{\bf {\large Amy Nejati}}}
\author{{{\bf {\large Abolfazl Lavaei}}}\vspace{0.4cm}\\
	{\normalfont School of Computing, Newcastle University, United Kingdom}\\
\texttt{\{b.samari2, amy.nejati, abolfazl.lavaei\}@newcastle.ac.uk}}

\pagestyle{fancy}
\lhead{}
\rhead{}
  \fancyhead[OL]{Behrad Samari, Amy Nejati, and Abolfazl Lavaei}

  \fancyhead[EL]{Model Order Reduction from Data with Certification}
  \rhead{\thepage}
 \cfoot{}
 
\begin{nouppercase}
	\maketitle
\end{nouppercase}

\section{Introduction}\label{Section: Introduction}
In recent years, there has been a growing inclination toward employing formal methods for the verification and controller synthesis of dynamical systems, aiming to enforce high-level logic properties, such as those outlined in linear temporal logic (LTL) formulas~\cite{pnueli1977temporal,baier2008principles}. Specifically, given a property of interest for a dynamical model, formal \emph{verification} aims to reliably assess whether the desired specification is fulfilled. As a more challenging endeavor, in a \emph{synthesis} problem involving dynamical models with control inputs, the primary objective is to formally design a controller, typically implemented as a state-feedback architecture, to enforce the desired property.  Providing formal verification and controller synthesis often poses inherent difficulties primarily due to computational challenges arising from the \emph{high dimensionality} of dynamical systems.

To alleviate this primary challenge, one promising approach is to employ model order reduction (MOR) techniques to construct a simplified \emph{lower-dimensional} system that accurately captures features of the original higher-order model. Subsequently, the constructed reduced-order model can serve as an effective substitute for the original system, facilitating analysis and synthesis over complex dynamical models. Specifically, results obtained from reduced-order models can be transferred to the original systems by designing an \emph{interface map} that accurately links the controllers of the original systems with those of their corresponding reduced-order models. By formally quantifying the disparity between output trajectories of the two systems using \emph{simulation functions}, it becomes feasible to ensure that the original systems also adhere to specifications akin to those of their ROM's counterparts within some guaranteed error bounds~\cite{antoulas2005approximation,lavaei2022automated}. 

While MOR techniques offer a valuable solution for formally analyzing high-dimensional systems, existing results often require \emph{precise knowledge} of original systems, which is not a common scenario in real-world applications~\cite{Hou2013model}. To address this challenge, \emph{indirect} data-driven approaches through identification techniques emerge to identify an approximate model of unknown systems~\cite{boots2007constraint,kolter2019learning,zhou2022neural}. However, this process involves complexities in two layers: (i) model identification, and (ii) solving the problem (\emph{i.e.,} constructing reduced-order models and simulation functions) using existing model-based approaches. Alternatively, and more intriguingly, there are \emph{direct} data-driven techniques that bypass the system identification phase and directly utilize observed data to analyze unknown systems.

\textbf{Primary contribution.} This work develops a data-driven methodology aimed at constructing reduced-order models from continuous-time control systems lacking  explicit mathematical descriptions. Central to our approach is the incorporation of \emph{simulation functions} to assess the closeness between output trajectories of the unknown original systems and their corresponding reduced-order models. Consequently, the resulting reduced-order models serve as effective counterparts within the controller synthesis procedure, with particular emphasis on fulfilling \emph{high-level logic} properties over unknown models. We introduce sufficient conditions for systematically constructing reduced-order models of unknown systems along with their associated simulation functions.  In particular, under a rank condition that can be fulfilled using data, we gather only \emph{two input-state trajectories} from unknown systems to construct both the reduced-order models and simulation functions, ensuring \emph{correctness guarantees}. We show the efficacy of our data-driven results across a set of benchmark scenarios encompassing various unknown systems while enforcing different complex properties, including \emph{safety} and \emph{reach-while-avoid} specifications.

\textbf{Related literature.} In the realm of systems and control, MOR techniques are commonly classified into three principal categories. The first category comprises energy-based methods, exemplified by \emph{balanced truncation}~\cite{moore1981principal,enns1984model,benner2011lyapunov} and optimal Hankel norm approximation~\cite{glover1984all}. Meanwhile, the second category encompasses Krylov methods, which rely on interpolation and/or \emph{moment matching} principles~\cite{feldmann1995efficient,grimme1997krylov}. The last category involves constructing reduced-order models through \emph{simulation functions} as a similarity relation between the original models and their reduced-order models \cite{zamani2017compositional,lavaei2017compositional,2016Murat,lavaei2020compositional,zhong2023automata,lavaei2019compositional}. 

Nevertheless, all these methodologies demand \emph{precise mathematical knowledge} of underlying systems. To address this restriction, several data-driven techniques have been introduced over the past few years within the first two categories. In particular, in the realm of energy-based techniques for linear systems, \cite{rapisarda2011identification} and \cite{markovsky2005algorithms} introduce a data-driven balanced truncation method derived from persistently exciting data. Moreover, \cite{burohman2023data} proposes a data-driven MOR scheme on the basis of noisy data with a known noise model. Leveraging the interpolatory methods, \cite{peherstorfer2017data} and \cite{drmavc2022learning} employ time-domain and noisy frequency-domain data, respectively, to formulate Loewner matrix pencils, facilitating the construction of state-space models. Within the second category, data-driven moment matching techniques are explored in \cite{astolfi2010model} and \cite{burohman2020data}, with the latter leveraging the data informativity framework. 

Our work is the first to offer data-driven techniques for constructing reduced-order models using \emph{simulation functions}, which is within the third category. Specifically, prior data-driven MOR studies within the first two categories (\emph{e.g.,} \cite{rapisarda2011identification,burohman2020data}) are primarily developed for \emph{stability} and \emph{input-state} behaviors, and cannot ensure the satisfaction of complex LTL properties such as safety, reachability, and reach-while-avoid. In contrast, our data-driven framework offers a guaranteed closeness between output trajectories of the original systems and those of the reduced-order models' counterparts, enabling the enforcement of complex properties beyond mere stability (cf. series of benchmarks).

\textbf{Organization.}
The rest of the paper is organized as follows. Section \ref{Section: Problem_Description} is allocated to presenting the mathematical preliminaries, notations, and formal definitions for continuous-time linear control systems and their corresponding reduced-order models, along with simulation functions. This section concludes by establishing the closeness guarantee between the two systems based on the simulation function. In Section \ref{Section: data-driven (linear) model order reduction}, we introduce our data collection approach and the data-driven construction of reduced-order models and simulation functions. We then extend the framework to address the verification problem. Section \ref{Simulation results} evaluates the efficacy of our proposed method through five benchmarks, and the paper concludes in Section \ref{conc}.

\section{Problem description}\label{Section: Problem_Description}

\subsection{Notation}\label{Subsection: Notation}
Sets of real, positive, and non-negative real numbers are denoted by $\mathbb{R},\mathbb{R}^+$, and $\mathbb{R}^+_0$, respectively. We signify the set of positive integers by $\mathbb{N}^+=\{1,2,\ldots\}$. An $n \times n$ identity matrix is represented by $\mathds{I}_n$. The zero vector of appropriate dimensions is denoted by $\boldsymbol{0}$.
Given $N$ vectors $\nu_i \in \mathbb{R}^{n}$, $\nu=[\nu_1 \, \, \ldots \,\, \nu_N]$ denotes the corresponding matrix of dimension $n \times N$. Euclidean norm of a vector $x\in\mathbb{R}^{n}$ is indicated by $\Vert x\Vert$.
We signify that a \emph{symmetric} matrix $P$ is positive definite by $P \succ 0$. Given a square matrix $P$, $\lambda_{\min}(P)$ and $\lambda_{\max}(P)$ denote minimum  and maximum eigenvalues of $P$, respectively. The \emph{right} pseudoinverse of a matrix $P$ is indicated by $P^\dagger$, while its range is represented by $\mathscr{R}(P)$.
We denote the supremum of a function $f: \mathbb{R}_0^+ \rightarrow \mathbb{R}^n$ by $\|f\|_{\infty}:=($\text{ess}$) \sup \{\|f(t)\|, t \geq 0\}$. A function $\beta: \mathbb{R}_{0}^+ \rightarrow \mathbb{R}_{0}^+$ is considered a $\mathcal{K}$ function if it is continuous, strictly increasing, and fulfills the condition $\beta(0) = 0$.
A function $\beta: \mathbb{R}_{0}^+ \times \mathbb{R}_{0}^+ \rightarrow \mathbb{R}_{0}^+$ is said to belong to class $\mathcal {KL}$ if, for each fixed $s$, the map $\beta(r,s)$ belongs to class $\mathcal K$ with respect to $r$, and for each fixed $r\in\mathbb{R}^+$\!, the map $\beta(r,s)$ is decreasing with respect to $s$, and $\beta(r,s)\rightarrow 0$ as $s \rightarrow \infty$.

\subsection{Original Systems and Reduced-Order Models}\label{Subsection: Continuous-Time Linear Systems}
We begin by introducing continuous-time linear control systems for which we aim to derive reduced-order models.

\begin{definition}\label{Def: ct-LCS model}
	A continuous-time linear control system (ct-LCS) is described by  
	\begin{align}\label{eq: ct-LCS model}
		\Sigma\!:\begin{cases}
			\dot x=Ax + Bu,\\
			y = x,
		\end{cases}
	\end{align}
	where $x \in X$ and $y  \in X$  are the system's state and output, $A \in \mathbb R^{n \times n}$ and $B \in \mathbb R^{n\times m}$ represent the system's matrices, and $u\in U$ denotes the control input, with $X\subset \mathbb R^{n}$, $U \subset \mathbb R^{m}$, and $X\subset \mathbb R^{n}$ being state, input, and output sets, respectively. The \emph{state trajectory} at time $t\in \mathbb{R}^+$ started from an initial condition $x_0 \in X$ under an input trajectory $u(\cdot)$ is denoted by $x_{x_0u}(t)$. The output value of $x_{x_0u}(t)$ is denoted by $y_{x_0u}(t)$, as the \emph{output trajectory} of ct-LCS. Since the output $y$ is an identity map of the state $x$, the state and output trajectories remain the same for ct-LCS in \eqref{eq: ct-LCS model}. We leverage the tuple $ \Sigma=(X, U, X,A,B,\mathds{I}_n)$ to refer to ct-LCS in \eqref{eq: ct-LCS model}.
\end{definition}
In our context, both matrices $A$ and $B$ are \emph{unknown}, reflecting real-world situations, and we use the term \emph{unknown model} to describe such systems. The subsequent definition formally characterizes the reduced-order model of ct-LCS in~\eqref{eq: ct-LCS model}.

\begin{definition}\label{Def: ct-LCS reduced-order model}
	A reduced-order model (ROM) of ct-LCS in~\eqref{eq: ct-LCS model} is described by
	\begin{align}\label{eq: ct-LCS reduced-order model}
		\hat\Sigma\!:\begin{cases}
			\dot{\hat x}=\hat A\hat x + \hat B\hat u,\\
			\hat y = \hat C \hat x,
		\end{cases}
	\end{align}
	where $\hat A \in \mathbb R^{\hat n \times \hat n}$, $\hat B \in \mathbb R^{\hat n\times \hat m}$, $\hat C \in \mathbb{R}^{n \times \hat n}$, potentially with $\hat n \ll n$, are ROM's matrices, and $\hat x \in \hat X$, $\hat u\in \hat U$, $\hat y \in \hat Y$, with $\hat X\subset \mathbb R^{\hat n}$, $\hat U \subset \mathbb R^{\hat m}$, $\hat Y \subset \mathbb{R}^{n}$ being the ROM's state, input, and output sets, respectively. We denote by $\hat x_{\hat x_0\hat u}(t)$ the ROM's \emph{state trajectory} at time $t\in \mathbb{R}^+$ under an input trajectory $\hat u(\cdot)$ starting from an initial condition $\hat x_0 \in \hat X$. We also denote by $\hat y_{\hat x_0\hat u}(t)$ the output value of $\hat x_{\hat x_0\hat u}(t)$, \textit{i.e.,} $\hat y_{\hat x_0\hat u}(t) = \hat C \hat x_{\hat x_0\hat u}(t)$, as the ROM's \emph{output trajectory}. We utilize the tuple $\hat\Sigma =(\hat X, \hat U, \hat Y,\hat A, \hat B, \hat C)$ to denote the ROM in \eqref{eq: ct-LCS reduced-order model}.
\end{definition}
Having defined the ct-LCS and its ROM counterpart, the following subsection presents the concept of simulation functions between a ct-LCS and its associated ROM. This notation primarily serves as an effective metric for assessing the closeness between output trajectories of ct-LCS and its ROM.

\subsection{Simulation Functions}

Simulation functions are Lyapunov-like functions, defining across the Cartesian product of state spaces, to capture the closeness between the trajectories of ct-LCSs and those of their ROMs. This notion of similarity ensures that any disparities between the two systems are constrained within a certain \emph{quantified error bound}.  Next, we provide a formal definition of simulation functions.

\begin{definition}\label{Def: Simulation Function-linear system}
	Consider a ct-LCS
	$\Sigma=(X, U, X,A,B,\mathds{I}_n)$ and its ROM
	$\hat\Sigma =(\hat X, \hat U, \hat Y,\hat A, \hat B, \hat C)$.
	A function $\mathcal V:X\times\hat X\to\R_0^+$ is
	called a simulation function (SF) from $\hat\Sigma$  to $\Sigma$, if there exist $\alpha,\kappa, \rho\in \R^+$ such that
	\begin{subequations}
		\begin{itemize}
			\item $\forall x\!\in\! X,\forall \hat x\!\in\!\hat X,$
			\begin{align}\label{eq:lowerbound2-linear}
				\alpha\Vert y-\hat y\Vert^2\le \! \mathcal V(x,\hat x),
			\end{align}
			\item $\forall x\in X,\forall\hat x\in\hat X, \forall \hat u \in \hat U, \exists u \in U,$ such that
			\begin{align}\label{eq:martingale2-linear}
				&\mathsf{L}\mathcal V(x,\hat x)\leq - \kappa \mathcal V(x,\hat x) + \rho \,\Vert \hat u \Vert^2\!,
			\end{align}
		\end{itemize}
	\end{subequations}
	where $\mathsf{L} \mathcal V$ is the Lie derivative of $\, \mathcal V:X\times\hat X\to\R_0^+$
	with respect to dynamics in~\eqref{eq: ct-LCS model} and \eqref{eq: ct-LCS reduced-order model}, defined as 
	\begin{align}\label{eq: Lie derivative-linear}
		\mathsf{L}\mathcal V(x,\hat x)=\partial_x \mathcal V(x,\hat x)(Ax + Bu) + \partial_{\hat x} \mathcal V(x,\hat x)(\hat A\hat x + \hat B\hat u),
	\end{align}
	where $\partial_x \mathcal V(x, \hat x) = \frac{\partial \mathcal{V}(x, \hat x)}{\partial x}$ and $\partial_{\hat x} \mathcal V(x, \hat x) = \frac{\partial \mathcal{V}(x, \hat x)}{\partial \hat x}$.
\end{definition}
The concept of SFs, as described in Definition \ref{Def: Simulation Function-linear system}, intuitively offers that if both a ct-LCS and its ROM start from nearby initial states (as ensured by condition~\eqref{eq:lowerbound2-linear}), their trajectories over time also remain close (as guaranteed by condition~\eqref{eq:martingale2-linear}) \cite{tabuada2009verification}.

\begin{remark}\label{interface}
	The second condition outlined in Definition \ref{Def: Simulation Function-linear system} essentially implies the presence of a function $u = u_{\hat u}(x, \hat x, \hat u)$ for the fulfillment of~\eqref{eq:martingale2-linear}. This function, known as the interface function, establishes a connection between control inputs $u, \hat u$, and serves to refine a synthesized control input $\hat{u}$ for $\hat{\Sigma}$ to a control input $u$ for $\Sigma$. We formally design this interface map using data in the next section, offering one of the primary contributions of our work (cf. Lemma \ref{Lemma-linear}).
\end{remark}

In the following theorem, we demonstrate the significance of an SF by quantifying the discrepancy between the output trajectories of $\Sigma$ and those of its ROM $\hat{\Sigma}$~\cite{zamani2017compositional}.

\begin{theorem}\label{thm-J19}
	Consider a ct-LCS $\Sigma=(X, U, X,A,B,\mathds{I}_n)$ as introduced in Definition~\ref{Def: ct-LCS model} and its ROM $\hat\Sigma =(\hat X, \hat U, \hat Y,\hat A, \hat B, \hat C)$ as in Definition~\ref{Def: ct-LCS reduced-order model}. Suppose $\mathcal{V}$ is an SF from $\hat\Sigma$  to $\Sigma$ as in Definition~\ref{Def: Simulation Function-linear system}. Then, there exists a $\mathcal{KL}$ function $\beta$ such that for any $x \in X, \hat x \in \hat X$, and $\hat{u} \in \hat U$, there exists $u \in U$, fulfilling the following closeness relation:
	\begin{align}
		\Vert y_{x_0u}(t) - \hat y_{\hat x_0\hat u}(t) \Vert &= \Vert x_{x_0u}(t) - \hat C \hat x_{\hat x_0\hat u}(t) \Vert \notag\\& \leq \frac{1}{\alpha}\,\beta\left(\mathcal{V}(x, \hat x), t\right) + \frac{\rho}{\alpha\kappa}\, \Vert \hat u \Vert_\infty. \label{new error}
	\end{align}
\end{theorem}

\begin{remark}
	The closeness relation \eqref{new error} is constructed using \cite[Theorem 3.3]{zamani2017compositional}, where $\alpha$, $\kappa$, and $\rho$ are all constants as in Definition~\ref{Def: Simulation Function-linear system}. It is also worth noting that if the initial conditions of the original system and its ROM are chosen such that $x_0 = \Theta \hat x_0$, the effect of the first term in \eqref{new error} is canceled, resulting in a smaller closeness error.
\end{remark}

\begin{remark}\label{Remark new}
	Since \eqref{new error} offers a closeness guarantee between the output trajectories of a ct-LCS and its ROM, the proposed results can be employed to enforce diverse complex properties beyond stability, such as safety, reachability, and reach-while-avoid \cite{tabuada2009verification}. Specifically, as properties of interest are typically defined over the output space of dynamical systems, ROMs can always be beneficial to enforce such properties over simplified lower-dimensional systems and refine the results back via an interface map across complex original systems while quantifying a guaranteed error bound between their closeness as in \eqref{new error}; thanks to the power of Theorem \ref{thm-J19}.
\end{remark}
To construct an SF and quantify the closeness between a ct-LCS and its ROM, as in \eqref{new error}, it is evident that knowledge of matrices $A$ and $B$ is necessary, given their appearance in the Lie derivative \eqref{eq: Lie derivative-linear}. Considering this primary challenge, we now proceed to formally define the problem targeted in this work.

\begin{resp}
	\begin{problem}\label{Problem: linear}
		Consider a ct-LCS in Definition~\ref{Def: ct-LCS model} with unknown matrices $A$ and $B$. Develop a data-driven framework with formal guarantees to construct a ROM for the ct-LCS, as in~\eqref{eq: ct-LCS reduced-order model}, an interface map   $u = u_{\hat u}(x, \hat x, \hat u)$ as in Remark \ref{interface}, and an SF between the ct-LCS and its ROM as in Definition \ref{Def: Simulation Function-linear system}. 
	\end{problem}
\end{resp}

To address Problem~\ref{Problem: linear}, we introduce our data-driven framework in the following section.

\section{Data-driven design of ROM and SF}\label{Section: data-driven (linear) model order reduction}
In our data-driven scheme, we first fix the structure of the SF to be quadratic in the form of $\mathcal V(x,\hat x) = (x - \Theta \hat x)^\top P(x - \Theta \hat x)$, where $P \succ 0$, and $\Theta \in \mathbb{R}^{n \times \hat n}$ is a \emph{reduction matrix}. We then collect input-state data from the unknown ct-LCS over the time interval $[t_0 , t_0 + (T - 1)\tau]$, where $T \in \mathbb N^{+}$ is the number of collected samples, and $\tau \in \mathbb R^{+}$ is the sampling time:
\begin{subequations}\label{single} 
	\begin{align}
		\mathcal U_{0,T} &= [u(t_0)~~u(t_0 + \tau)~~\dots~~u(t_0 + (T - 1)\tau)],\label{U0}\\
		\mathcal X_{0,T} &= [x(t_0)~~x(t_0 + \tau)~~\dots~~x(t_0 + (T - 1)\tau)],\label{X0}\\
		\mathcal X_{1,T} &= [\dot x(t_0)~~\dot x(t_0 + \tau)~~\dots~~\dot x(t_0 + (T - 1)\tau)].\label{X1}
	\end{align}
\end{subequations}
Likewise, we gather an additional trajectory through the ct-LCS spanning the identical time interval $\bar{\mathcal X}_{0,T} = [x(t_0)~~x(t_0 + \tau)~~\dots~~x(t_0 + (T - 1)\tau)]$ and $\bar{\mathcal X}_{1,T} = [\dot x(t_0)~~\dot x(t_0 + \tau)~~\dots~~\dot x(t_0 + (T - 1)\tau)]$ with a constantly-zero control input $\bar{\mathcal U}_{0,T} = [\boldsymbol{0}~~\boldsymbol{0}~~\dots~~\boldsymbol{0}]$. It is of vital importance to note that both trajectories should start from the same initial condition. We consider the trajectories in \eqref{single} and their zero-input counterparts as \emph{two input-state trajectories}.

\begin{remark}\label{remark: deriv}
	Data $\dot{x}\left(t_0 + k\tau\right)$, $k \in \{0, 1, \dots, (T-1)\}$, in $\mathcal{X}_{1, T}$ and $\bar{\mathcal{X}}_{1,T}$, can be approximated utilizing numerical differentiation. For instance, if the forward difference approximation method is employed, we have $\dot{x}_i\left(t_0+k\tau\right)=\frac{x_i\left(t_{0}+(k+1)\tau\right)-x_i\left(t_0+k\tau\right)}{\tau}+e_i\left(t_0+k\tau\right), \; i\in\{1, \ldots, n\},$
	where the approximation error $e_i\left(t_0+k\tau\right)$ is proportional to $\tau$ and can be considered as noise~\cite{guo2022data}. For the sake of clarity, this work assumes that the data is free of noise; however, if the data is corrupted by noise, numerical techniques such as total variation regularization can be employed to minimize the impact of noise in the derivative approximation~\cite{rudin1992nonlinear}. It is worth noting that if the unknown system evolves in discrete time, the input-state trajectory over the interval $[0,T-1]$, where $T\in\mathbb{N}^+,$ will be
	\begin{align*}
		\mathcal U_{0,T}&=[u(0)\;u(1)\;\dots\;u(T-1)],\\
		\mathcal X_{0,T}&=[x(0)\;x(1)\;\dots\;x(T-1)],\\
		\mathcal X_{1,T}&=[x(1)\;x(2)\;\dots\;x(T)].
	\end{align*}
	This means the data of the vector field (a.k.a. transition function in the discrete-time domain) is one-step forward of the states' data, which is error-free.
\end{remark}

We now offer the following lemma to obtain the data-based representation of the ct-LCS in~\eqref{eq: ct-LCS model}. This enables the characterization of both \emph{unknown matrices} $A$ and $B$ exclusively through the trajectories of the unknown ct-LCS while designing an interface map purely based on data.

\begin{lemma}\label{Lemma-linear}
	Let $Q$ and $\bar Q$ be $(T\times n)$ matrices such that
	\begin{subequations}\label{Q_new}
		\begin{align}
			\mathds{I}_n = \mathcal X_{0,T}Q,
			\label{eq: I_N=X_0 Q}\\
			\mathds{I}_n =\bar{ \mathcal X}_{0,T}\bar Q,\label{eq: I_N=X_0 Qbar}
		\end{align}
	\end{subequations}
	with $\mathcal X_{0,T}$ and $\bar{ \mathcal X}_{0,T}$ being $(n\times T)$ full row-rank matrices.
	By designing an interface map
	\begin{align}\label{interface1}
		u = F(x-\Theta\hat x)+\Xi \hat x + \Psi \hat u,
	\end{align}
	where $F = \mathcal U_{0,T}Q$, $\Xi \in \mathbb R^{m \times \hat n}$, and $ \Psi \in \mathbb R^{m \times \hat m}$, the closed-loop system has the following data-based representation:
	\begin{align}
		\dot x= \mathcal X_{1,T}Qx  - \left(\mathcal X_{1, T} Q - \bar{\mathcal{X}}_{1, T} \bar Q\right)\! \Theta\hat x  + \myexpression \! .\label{eq: data-dotx-linear}
	\end{align}
\end{lemma}

\begin{proof}
	By leveraging the proposed interface map $u = F(x-\Theta\hat x)+\Xi \hat x + \Psi \hat u = \mathcal U_{0,T}Q (x-\Theta\hat x)+\Xi \hat x + \Psi \hat u$, with $F = \mathcal U_{0,T}Q$, we have
	\begin{align}\notag
		\dot x = Ax + Bu &= Ax +BF(x - \Theta\hat x)+B\Xi \hat x + B\Psi\hat u\\\notag
		&=Ax +BFx - BF\Theta\hat x+B\Xi \hat x + B\Psi\hat u\\\notag
		& = \underbrace{(A\overbrace{\mathcal{X}_{0, T}Q}^{\mathds{I}_n} + B\overbrace{\mathcal{U}_{0,T}Q}^{F})}_{A + BF}x - BF\Theta\hat x+B\Xi \hat x + B\Psi\hat u\\\notag
		& = \underbrace{(A \mathcal{X}_{0, T} + B\mathcal{U}_{0,T})}_{\mathcal{X}_{1, T}}Qx - BF\Theta\hat x+B\Xi \hat x + B\Psi\hat u\\\label{closed}
		&= \mathcal{X}_{1, T} Q x - BF\Theta\hat x+B\Xi \hat x + B\Psi\hat u .
	\end{align}
	Now, in order to obtain the data-based representation of matrix $B$, we use the \emph{second input-out trajectory} collected with the zero control input, \emph{i.e.,} $\bar{\mathcal U}_{0,T} = [\boldsymbol{0}~~\boldsymbol{0}~~\dots~~\boldsymbol{0}]$. Thus, we have
	\begin{align}\label{A}
		\dot x = Ax &= A\underbrace{\bar{\mathcal{X}}_{0,T}\bar Q}_{\mathds{I}_n}x= \bar{\mathcal{X}}_{1, T} \bar Q x,
	\end{align}
	with $ \bar{\mathcal{X}}_{1, T}  = A\bar{\mathcal{X}}_{0,T}$. From \eqref{A}, one can conclude that the data-based representation of matrix $A$ is $ \bar{\mathcal{X}}_{1, T} \bar Q$. On the other hand, according to \eqref{eq: I_N=X_0 Q} and since $F = \mathcal U_{0,T}Q$, we have 
	\begin{align}\label{new}
		A + BF = A + B ~\!\mathcal U_{0,T}Q = \underbrace{[B\quad A] \begin{bmatrix}
				\mathcal U_{0,T}\\
				\mathcal X_{0,T}
		\end{bmatrix}}_{\mathcal{X}_{1, T}}\!Q  = \mathcal{X}_{1, T} Q.
	\end{align}
	By leveraging both \eqref{A} and \eqref{new}, we can represent matrix $B$ based solely on trajectories as follows:
	\begin{align*}
		\underbrace{A}_{\bar{\mathcal{X}}_{1, T} \bar Q} +~~ B\!\!\!\underbrace{F}_{\mathcal{U}_{0,T}Q } \!=~\mathcal{X}_{1, T}Q~ & \to ~
		B\mathcal{U}_{0,T}Q \!=\! \mathcal{X}_{1, T}Q - \bar{\mathcal{X}}_{1, T} \bar Q ~  \to ~
		B = \Bdata \!.
	\end{align*}
	Now, one can represent the closed-loop ct-LCS in \eqref{closed}  based on data as
	\begin{align*}
		\dot x = \mathcal X_{1,T}Qx & - \left(\mathcal X_{1, T} Q - \bar{\mathcal{X}}_{1, T} \bar Q\right)\! \Theta\hat x  + \myexpression \!,
	\end{align*}
	which completes the proof. 
\end{proof}

\begin{remark}\label{remark: data}
	To ensure both $\mathcal X_{0,T}$ and $\bar{\mathcal X}_{0,T}$ are full-row rank matrices— a requirement for designing $Q$ and $\bar{Q}$ in \eqref{Q_new}—the number of samples $T$ must be greater than $n$. Given that matrices $\mathcal X_{0,T}$ and $\bar{\mathcal X}_{0,T}$ are both constructed from sampled data, this assumption is readily validated.
\end{remark}

By employing the proposed data-based representation as in Lemma~\ref{Lemma-linear}, we offer the following theorem, as the main result of the work, to establish an SF from data while constructing the ROM $\hat \Sigma$.

\begin{theorem}\label{Thm:main1}
	Consider the unknown ct-LCS  \eqref{eq: ct-LCS model}, its closed-loop data-based representation \eqref{eq: data-dotx-linear}, and let $\hat C=\Theta.$ If there exist matrices $\mathcal H \in \mathbb R^{T \times n}, \hat{A}\in \mathbb{R}^{\hat n\times\hat n}, \Xi \in \mathbb{R}^{m \times \hat n},$ and $\Theta \in \mathbb{R}^{n \times \hat n},$ for some $\hat{\kappa} \in \mathbb{R}^+$, such that
	\begin{subequations}\label{eq: eq_thm}
		\begin{align}
			\label{eq: P-linear}
			&\mathcal X_{0,T}\mathcal H = P^{-1}, \quad \text{with} ~ P \succ 0,\\
			\label{eq: conditions-linear3}
			&\bar{\mathcal{X}}_{1, T} \bar Q \Theta = \Theta\hat A - \Bdata \Xi ,\\\label{eq: condition1-linear3}
			& \mathcal H^\top \mathcal X_{1,T}^\top +  \mathcal X_{1,T}\mathcal H  \preceq -\hat \kappa \mathcal X_{0,T}\mathcal H,
		\end{align}
	\end{subequations}
	then $\mathcal V(x,\hat x) = (x- \Theta\hat x)^\top P (x- \Theta\hat x)$ is an SF from $\hat\Sigma$  to $\Sigma$, with  $\rho=  \dfrac{1}{\varepsilon}\big\Vert \! \, \sqrt{P} \big(\!\!\Bdata \Psi - \Theta \hat B \big) \big\Vert^2$, $\alpha=\lambda_{\min}(P)$, and $\kappa = \hat \kappa - \varepsilon$, where $ 0<  \varepsilon<\hat \kappa$. In addition, the controller is in the form of \eqref{interface1}, and $\hat A\hat x + \hat B \hat u$ is the ROM for the unknown ct-LCS with matrix $\hat A$ obtained from equality condition \eqref{eq: conditions-linear3} and any arbitrary matrix $\hat B$ of an appropriate dimension.
\end{theorem}

\begin{proof}
	We first show that condition~\eqref{eq:lowerbound2-linear} holds. Given that $\hat C = \Theta$, we have 
	\begin{align}\label{New}
		\Vert y -\hat y\Vert^2 = \Vert x -\hat C\hat x\Vert^2 = \Vert x -\Theta\hat x\Vert^2. 
	\end{align}
	Since $\lambda_{\min}(P)\Vert x- \Theta\hat x\Vert^2\leq\overbrace{(x-\Theta\hat x)^\top P(x-\Theta\hat x)}^{\mathcal V(x,\hat x)}\leq\lambda_{\max}(P)\Vert x- \Theta\hat x\Vert^2$ and according to~\eqref{New}, it can be readily verified that $\lambda_{\min}(P)	\Vert y -\hat y\Vert^2 = \lambda_{\min}(P)\Vert x- \Theta\hat x\Vert^2\le \mathcal V(x,\hat x)$ holds for any $x \in X$ and $\hat x \in \hat X$, implying that inequality~\eqref{eq:lowerbound2-linear} holds with $\alpha=\lambda_{\min}(P)$. 
	
	We now proceed with showing condition~\eqref{eq:martingale2-linear}, as well. Since $\mathcal X_{0,T}\mathcal H  P = \mathds{I}_n $ from~\eqref{eq: P-linear} and $\mathcal X_{0,T}Q = \mathds{I}_n$ from \eqref{eq: I_N=X_0 Q}, one can conclude that $Q = \mathcal HP$ and, accordingly, $Q P^{-1} = \mathcal H$. Given that $A + BF = \mathcal X_{1,T}Q$ according to Lemma~\ref{Lemma-linear}, we have
	\begin{align}\label{new2}
		\big(A + BF\big)P^{-1} = \mathcal X_{1,T}\underbrace{QP^{-1}}_{\mathcal H}  = \mathcal X_{1,T}\mathcal H.
	\end{align}
	Now to show condition~\eqref{eq:martingale2-linear}, using the definition of Lie derivative in \eqref{eq: Lie derivative-linear}, we have
	\begin{align*}
		\mathsf{L}\mathcal V(x,\hat x)&=\partial_x \mathcal V(x,\hat x)\big(Ax + Bu\big) + \partial_{\hat x} \mathcal V(x,\hat x)\big(\hat A\hat x + \hat B \hat u\big)\\
		&= 2(x - \Theta\hat x)^\top P \big(Ax + Bu\big)  - 2(x - \Theta\hat x)^\top P\Theta \big(\hat A\hat x + \hat B \hat u\big)\\
		&=2(x - \Theta\hat x)^\top P\Big[Ax + Bu - \Theta \hat A \hat x - \Theta \hat B \hat u\Big]\!.
	\end{align*}
	By substituting the data-based representation of $Ax + Bu$ in~\eqref{eq: data-dotx-linear}, one has
	\begin{align*}
		\mathsf{L}\mathcal V(x,\hat x) = 2(x - \Theta\hat x)^\top P \! \Big[  \mathcal X_{1,T}Qx  \!-\! \left(\mathcal X_{1, T} Q \!-\! \bar{\mathcal{X}}_{1, T} \bar Q\right)\Theta\hat x \!  +\! \myexpression  \!-\! \Theta\hat A \hat x \!-\! \Theta\hat B \hat u\Big]\!.
	\end{align*}
	Then, according to the proposed equality condition~\eqref{eq: conditions-linear3}, one can obtain
	\begin{align}
		\mathsf{L}\mathcal V(x,\hat x)= &\:\:2(x - \Theta\hat x)^\top P\mathcal{X}_{1, T}Q(x-\Theta\hat x)+ 2(x - \Theta\hat x)^\top P \Big[\!\!\Bdata \Psi - \Theta \hat B\Big]  \hat u. \label{new8}
	\end{align}
	Now, according to the Cauchy-Schwarz inequality \cite{bhatia1995cauchy}, \emph{i.e.,}  $a b \leq \vert a \vert \vert b \vert,$ for any $a^\top, b \in \R^{n}$, followed by
	employing Young's inequality \cite{young1912classes}, \emph{i.e.,} $\vert a \vert \vert b \vert \leq \frac{\varepsilon}{2} \vert a \vert^2 + \frac{1}{2\varepsilon} \vert b \vert^2$, for any $\varepsilon >0$, over the second term of \eqref{new8}, one has
	\begin{align}\notag
		& 2(x - \Theta\hat x)^\top \underbrace{P}_{\sqrt{P} \sqrt{P}}  \Big[ \Bdata \Psi - \Theta \hat B \Big]\hat u \\ & \leq \varepsilon \overbrace{(x - \Theta\hat x)^\top P (x - \Theta\hat x)}^{\mathcal V(x,\hat x)} +\frac{1}{\varepsilon}\big\Vert  \sqrt{P} \big(\!\!\Bdata \Psi - \Theta \hat B \big)\hspace{-0.04cm} \big\Vert^2 \Vert \hat u \Vert^2. \label{new1}
	\end{align}
	Thus by employing \eqref{new1}, we have
	\begin{align*}
		\mathsf{L}\mathcal V(x,\hat x) & \leq  2(x - \Theta\hat x)^\top \big[P \mathcal{X}_{1, T}Q \big](x - \Theta\hat x) + \varepsilon \mathcal V(x,\hat x)\\& \quad + \frac{1}{\varepsilon}\big\Vert \! \, \sqrt{P} \big(\!\!\Bdata \Psi - \Theta \hat B \big) \big\Vert^2 \Vert \hat u \Vert^2\\
		&\leq (x - \Theta\hat x)^\top \Big [Q^\top \mathcal{X}_{1, T}^\top P + P \mathcal X_{1, T}Q\Big](x - \Theta\hat x) + \varepsilon \mathcal V(x,\hat x) \\&\quad+ \frac{1}{\varepsilon}\big\Vert \! \, \sqrt{P} \big(\!\!\Bdata \Psi - \Theta \hat B \big) \big\Vert^2 \Vert \hat u \Vert^2\\
		& = (x - \Theta\hat x)^\top P\Big [P^{-1}(\mathcal X_{1, T}Q)^\top + (\mathcal X_{1, T}Q)P^{-1} \Big]P(x - \Theta\hat x) + \varepsilon \mathcal V(x,\hat x)\\
		&\quad + \frac{1}{\varepsilon}\big\Vert \! \, \sqrt{P} \big(\!\!\Bdata \Psi - \Theta \hat B \big) \big\Vert^2 \Vert \hat u \Vert^2.
	\end{align*}
	Since $\underbrace{\big(A + BF\big)}_{\mathcal{X}_{1, T}Q}P^{-1} = \mathcal X_{1,T}\mathcal H$ according to \eqref{new2}, and considering condition \eqref{eq: condition1-linear3}, one has
	\begin{align*}
		\mathsf{L}\mathcal V(x,\hat x)&\leq (x - \Theta\hat x)^\top P\Big [\mathcal H^\top \mathcal X_{1,T}^\top +  \mathcal X_{1,T}\mathcal H \Big]P(x - \Theta\hat x) + \varepsilon \mathcal V(x,\hat x)\\
		&\quad + \frac{1}{\varepsilon}\big\Vert \!\! \, \sqrt{P} \big(\!\!\Bdata \Psi - \Theta \hat B \big) \big\Vert^2 \Vert \hat u \Vert^2\\
		&\leq -\hat \kappa (x - \Theta\hat x)^\top P\overbrace{\mathcal X_{0,T}\mathcal H P}^{\mathds{I}_n}(x - \Theta\hat x)^\top + \varepsilon \mathcal V(x,\hat x)\\&\quad + \frac{1}{\varepsilon}\big\Vert \! \, \sqrt{P} \big(\!\!\Bdata \Psi - \Theta \hat B \big) \big\Vert^2 \Vert \hat u \Vert^2\\&= - \kappa \mathcal V(x,\hat x) + \rho\Vert \hat u\Vert^2,
	\end{align*}
	with $\rho=  \dfrac{1}{\varepsilon}\big\Vert \!\! \, \sqrt{P} \big(\!\!\Bdata \Psi - \Theta \hat B \big) \big\Vert^2$ and $\kappa = \hat \kappa - \varepsilon$, where $ \hat \kappa > \varepsilon$, thus completing the proof.
\end{proof}

\begin{remark}\label{remark1}
	The equivalence of $\mathcal X_{1,T}\mathcal H$ in~\eqref{eq: condition1-linear3} to $\big(A + BF\big)P^{-1}$ in \eqref{new2} indicates that ~\eqref{eq: condition1-linear3} is feasible if and only if the pair $(A, B)$ is stabilizable. This can be easily verified using semi-definite programming (SDP) tools like \textsf{SeDuMi}~\cite{sturm1999using} and \textsf{Mosek}\footnote{Mosek license was obtained from \href{https://www.mosek.com}{https://www.mosek.com}.}. While both \textsf{SeDuMi} and \textsf{Mosek} perform well for systems with a relatively small number of states (e.g., $n < 10$), \textsf{Mosek} typically demonstrates superior performance when handling high-dimensional systems.
\end{remark}

\begin{remark}\label{hatB}
	Since Theorem~\ref{Thm:main1} does not put any restriction on the matrix $\hat B$, it is highly recommended to choose $\hat B=\mathds{I}_{\hat n}$ to ensure the ROM $\hat \Sigma$ is fully actuated, thereby simplifying the synthesis problem over $\hat \Sigma$. If this choice results in a relatively large value for $\rho$, an alternative is to use a scaled identity matrix to regulate $\rho$ and the error accordingly (see Subsection \ref{subsec4.2} where $\hat{B}=0.1\times \mathds{I}_2$).
\end{remark}

Given that Theorem~\ref{Thm:main1} does not place any specific condition on the matrix $\Psi$ in the interface function~\eqref{interface1}, it is advisable to select $\Psi$ {in a way that minimizes $\rho$, as discussed in~\cite{girard2009hierarchical} for the model-based setting. The choice of $\Psi = \Psi_{1} \Psi_{ 2}$ with the following matrices has the potential to minimize $\rho$, and accordingly, minimize the closeness between the trajectories of the ct-LCS and its ROM as in \eqref{new error}:
	\begin{subequations}\label{matrix psi}
		\begin{align}
			\Psi_{1} &= \Big(\!\big(\!\!\Bdata\!\big)^{\!\top} \! P \big(\!\!\Bdata\!\big)\!\!\hspace{0.01cm}\Big)^{\! -1}\! \!\!\!\!,\\
			\Psi_{2} &= \big(\!\!\Bdata\!\big)^{\!\top}   P \Theta \hat B.
		\end{align}
	\end{subequations}
	It is worth highlighting that the design choice of the interface map, \textit{i.e.,} $u = \mathcal U_{0,T}Q (x-\Theta\hat x)+\Xi \hat x + \Psi \hat u$, as proposed in \eqref{interface1}, offers two key benefits. Firstly, the term $\Xi \hat x$ in the interface leads to the appearance of $\Xi$ in~\eqref{eq: conditions-linear3} as a decision variable, enhancing the satisfaction of this condition by giving the solver an additional degree of freedom. Secondly, the term $\Psi \hat u$ enables one to achieve the minimum value for $\rho$, as discussed earlier, thereby minimizing the closeness between trajectories of the ct-LCS and its ROM.
	
	As detailed in Remark \ref{remark1}, stabilizability of the pair $(A, B)$ is sufficient and necessary for the satisfaction of condition \eqref{eq: condition1-linear3}. As for condition \eqref{eq: conditions-linear3}, we present the following lemma, which provides a sufficient and necessary \emph{geometric condition} for its fulfillment~\cite{zamani2017compositional}.
	\begin{lemma}\label{lemma:chooseA}
		There exist matrices $\hat A$ and $\Xi$ satisfying equality condition~\eqref{eq: conditions-linear3} if and only if
		\begin{align*}
			\mathscr{R}(\bar{\mathcal{X}}_{1, T} \bar Q \Theta) \subseteq \mathscr{R}(\Theta) + \mathscr{R}\big(\!\!\Bdata\!\big),
		\end{align*}
		with $\mathscr{R}$ denoting the range of underlying matrices.
	\end{lemma}
	
	We provide Algorithm \ref{Alg:1} to summarize the required steps for constructing a ROM and its SF from data with provable guarantees.
	
	\subsection{Verification problem}
	If the original system does not include any control input, the problem of controller synthesis is simplified to a verification task, which is a subclass of the synthesis problem. In such scenarios, there is no need for an interface map design due to lack of inputs, and the primary task is to construct matrices $\hat A$ and $\hat C$, while establishing an SF between the trajectories of two systems. Consequently,  $\hat u = \boldsymbol{0}$ in \eqref{eq:martingale2-linear}, and conditions \eqref{eq: conditions-linear3}-\eqref{eq: condition1-linear3} are simplified to the following two conditions:
	\begin{subequations}
		\begin{align*}
			&\bar{\mathcal{X}}_{1, T} \bar Q \Theta = \Theta\hat A,\quad \mathcal H^\top \bar{\mathcal{X}}_{1,T}^\top +  \bar{\mathcal{X}}_{1,T}\mathcal H  \preceq -\hat \kappa \bar{\mathcal{X}}_{0,T}\mathcal H.
		\end{align*}
	\end{subequations}
	
	\begin{algorithm}[t!]
		\caption{Data-driven MOR and SF construction with provable guarantees}\label{Alg:1}
		\begin{center}
			\begin{algorithmic}[1]
				\REQUIRE 
				Desired degree of ROM $\hat{\Sigma}$ (\textit{i.e.,} $\hat{n}$) 
				\STATE 
				Gather input-state trajectories $\mathcal{U}_{0, T}, \mathcal{X}_{0, T}, \mathcal{X}_{1, T}, \bar{\mathcal{X}}_{0, T},$ and $\bar{\mathcal{X}}_{1, T}$ according to \eqref{U0}-\eqref{X1}
				\STATE
				Compute $\bar{Q}$ as $\bar{Q} = \bar{\mathcal{X}}_{0, T}^\dagger$ according to \eqref{eq: I_N=X_0 Qbar}
				\STATE
				For a fixed choice of $\hat{\kappa}$, solve \eqref{eq: condition1-linear3} subject to \eqref{eq: P-linear}\footnotemark, and obtain $\mathcal{H}$ and $P \succ 0$
				\STATE
				Compute $Q$ as $Q = \mathcal{H}P$, which satisfies \eqref{eq: I_N=X_0 Q} automatically
				\STATE
				Set $\hat{A}$ to be a Hurwitz matrix; then, solve~\eqref{eq: conditions-linear3} to obtain $\Xi$ and $\Theta$
				\STATE
				Set $\hat{B}=\mathds{I}_{\hat n}$ according to Remark \ref{hatB}
				\STATE
				Compute $\hat C = \Theta$ as per Theorem \ref{Thm:main1}
				\STATE
				Compute $\Psi$, as part of the interface map, according to \eqref{matrix psi}
				\STATE
				Compute $\alpha, \kappa, \rho$ according to Theorem \ref{Thm:main1}
				\STATE
				Quantify the closeness between output trajectories of ct-LCS $\Sigma$ and ROM $\hat{\Sigma}$ according to \eqref{new error}\vspace{-0.3cm}
				\ENSURE
				Simulation function $\mathcal V(x,\hat x) = (x- \Theta\hat x)^\top \overbrace{[\mathcal X_{0,T}\mathcal H ]^{-1}}^P(x- \Theta\hat x)$, ROM's matrices  $\hat A, \hat B, \hat C$, control input $u = \mathcal U_{0,T}Q (x-\Theta\hat x)+\Xi \hat x + \Psi \hat u$, and guaranteed closeness quantification between trajectories of ct-LCS and ROM according to~\eqref{new error}
			\end{algorithmic}
		\end{center}
	\end{algorithm}
	\footnotetext{We solve \eqref{eq: condition1-linear3} while considering~\eqref{eq: P-linear}, ensuring that the resulting matrix from the product $\mathcal X_{0,T}\mathcal H$ is a symmetric, positive-definite matrix. Subsequently, we interpret the resultant matrix of this product as $P^{-1}$, and through the matrix inversion process, we derive the matrix $P$ itself.}
	
	\subsection{Limitations}\label{LIMITATIONS}
	Similar to any approach, our findings have some limitations that can be examined from different aspects. First and foremost, our framework is developed at this stage solely for linear systems. Extending our results to encompass certain classes of nonlinear systems would be of great interest to include a broader range of systems. It is important to highlight that this work is the first to address the MOR problem from a data perspective using SFs, and it is considered a fundamental result that will serve as a foundation for future developments in this area. For instance, building upon this work, we are exploring its extension to a class of nonlinear dynamical systems, structured as
	\begin{align}\label{Nonlinear}
		\begin{cases}\dot{x}=A\mathcal{Z}(x)+Bu,\\y=x,\end{cases}
	\end{align}
	where $A \in \mathbb{R}^{n\times Z}$, and $\mathcal{Z}:X\rightarrow\mathbb{R}^{Z}$ is the continuous nonlinear basis functions with $x,X,y,B,$ and $u$ as introduced in Definition \ref{Def: ct-LCS model}. In this case, in addition to developing the corresponding conditions, the interface function should be designed as $u=\mathcal U_{0,T}Q(\mathcal Z(x)-\Theta\hat x)+\Xi\hat x+\Psi\hat u$, which is a nonlinear function, with $Q\in\mathbb R^{T\times Z},\Theta\in\mathbb R^{Z\times\hat n}$.
	
	Secondly, $\mathcal{X}_{1, T}$ and $\bar{\mathcal{X}}_{1, T}$ cannot be directly measured as they include the states' derivatives at sampling times. To address this, one possible solution is to use appropriate filters to approximate the derivatives while considering the approximation error~\cite{padoan2015towards} (cf. Remark~\ref{remark: deriv}). Alternatively, to directly incorporate the effect of noise in the derivative approximation into the framework, a method similar to that in~\cite{guo2021data} can be used, where $\mathcal X_{1,T}=\widehat{\mathcal X}_{1,T}+D$, with $\widehat{\mathcal X}_{1,T}$ as the error-free data and $D$ as the noise capturing the underlying error (this should also be applied to the second trajectory). The only assumption required is that $D$ satisfies $DD^\top\preceq WW^\top$ for some known $W$, intuitively meaning that the energy of the noise during data gathering is bounded. We opted not to use this method in the paper for the sake of a clearer presentation.
	
   Ultimately, to utilize our framework, one needs to solve~\eqref{eq: I_N=X_0 Qbar}, which requires collecting a trajectory of the system (\emph{i.e.,} $\bar{\mathcal{X}}_{0, T}$) with no control input (\emph{i.e.,} $u = \boldsymbol{0}$). Solving~\eqref{eq: I_N=X_0 Qbar} necessitates that $\bar{\mathcal{X}}_{0, T}$ be a full-row rank matrix (see Remark~\ref{remark: data}), implying that matrix $A$ of the system should be full rank (since there is no control input when collecting data). However, this may not be the case in some applications. To overcome this difficulty in such scenarios, one needs to know the knowledge of matrix $B$, which is not a particularly restrictive assumption. In this case, there is no longer a need to solve~\eqref{eq: I_N=X_0 Qbar}, and conditions~\eqref{eq: conditions-linear3}-\eqref{eq: condition1-linear3} can be adapted accordingly so that matrix $B$ appears in those conditions instead of its data-based representation.
	
	\subsection{Discussions}
	
	Matrix $\mathcal{H}$ is an auxiliary matrix upon which we prevent inequality \eqref{eq: condition1-linear3} from being bilinear. Furthermore, $\varepsilon$ is also a design parameter that can help decline $\rho$, and the only assumption on this parameter is to be lower than $\hat{\kappa}$. However, the bigger this parameter is, the lower $\rho$ will be, so we desire to obtain it as large as possible. Moreover, $\hat{\kappa}$ is also a design parameter; the bigger this parameter is, the lower the guaranteed error bound will be. However, satisfying \eqref{eq: condition1-linear3} might be challenging if it is very big. Thus, there should be a trade-off in designing these parameters. Finally, $\kappa$ in \eqref{eq:martingale2-linear}, the decay rate of the SF, would be obtained as $\kappa=\hat{\kappa}-\varepsilon$.
	
	\noindent {\bf Errors of SDP Solvers.} While numerical errors are inherent in SDP problems, they are usually on the order of $10^{-8}$ or lower, especially for linear systems, making SDP solvers widely useful in control literature. These minor errors are typically negligible since control strategies are often robust to small perturbations. In our framework, the positive definiteness of matrix $P$  is also expected to be maintained despite minor errors.
	
	\noindent {\bf Practical Scalability.} Our framework is general and does not impose any restriction on the system's state dimension for which we aim to construct a ROM. A potential restriction, however, could be posed by the SDP solvers at hand. Specifically, the computational complexity of solving LMI~\eqref{eq: condition1-linear3} using \textsf{SeDuMi} is in $O(T^2n^{4.5}+n^{3.5})$, with $n$ being the state dimension and $T$ the number of collected samples. To address this and enhance scalability, high-dimensional systems (e.g., those with $100$ states) could be approached as networks that can be decoupled into subsystems, rather than solving the problem as a single entity, as a future research direction.
	
	\begin{table*}[t!]
		\caption{A brief summary of our data-driven findings across a set of benchmarks.}
		\label{Table}
		\centering
		\begin{tabular*}{\linewidth}{@{\extracolsep{\fill}}lcclcc}
			\toprule
			System & $n$ & $\hat{n}$ & Specification & Number of data $T$ & Run time (sec)\\
			\midrule
			Motor control system & $5$ & $1$ & Safety & $6$ & 0.10\\
			Spacecraft & $4$ & $2$ & Reach-while-avoid & $10$ & 0.11\\
			Blood glucose metabolism & $7$ & $1$ & Safety & $250$ & 0.52\\
			Cart system & $6$ & $1$ & Tracking & $7$ & 0.17\\
			High-dimensional system &  $25$ &  $2$ &  Reach-while-avoid&  $48$ &  0.68\\
			\bottomrule
		\end{tabular*}
	\end{table*}
	
	\section{Set of benchmarks}\label{Simulation results}
	We show the efficacy of our data-driven MOR findings by applying them to five benchmarks, including \emph{a motor control system, a spacecraft, a human blood glucose metabolism, a cart system with a double-pendulum controller, and a system with 25 states}, the latter of which specifically demonstrates the applicability of our results to high-dimensional systems. \emph{The mathematical models of all case studies are assumed to be unknown to us}, and the main goal is to construct ROMs for these unknown systems together with their corresponding SFs. The data-driven ROMs are then utilized for designing controllers that enforce various properties across original systems, including \emph{tracking, safety, and reach-while-avoid specifications}. The general yet brief information of these case studies is provided in Table~\ref{Table}.
	
	All simulations were performed in MATLAB on a MacBook Pro (Apple M2 Max with 32GB memory). Below, we provide the details of each case study.
	
	\subsection{Motor control system}\label{cSCo}
	As for the first case study, we apply our data-driven MOR findings to a 5-dimensional motor control system~\cite{tran2016large}. Matrices $A$ and $B$ of the system are as follows:
	\begin{align*}
		A =& \begin{bmatrix}
			-18.925 & 0 & 0.22572 & 80.823 & 29.973\\ 0 & -18.925 & -80.823 & 0.22572 & 0.017978\\ -0.22572 & 80.823 & -76.569 & 0 & -0.26957\\ -80.823 & -0.22572 & 0 & -76.569 & -122.93\\ 29.973 & 0.017978 & 0.26957 & 122.93 & -194.95
		\end{bmatrix}\!\!,\\
		B =& \begin{bmatrix}
			5.3423 & 0 & 0.02 & 6.8169 & 4.5252\\ 0 & -5.3423 & 6.8169 & -0.02 & 0.003
		\end{bmatrix}^{\!\top}\!\!\!\!\!,
	\end{align*}
	but we stress the fact that we assume these matrices are unknown to us.
	
	The primary objective is to reduce the dimension of the system from five to one (\emph{i.e.,} from $n=5$ to $\hat n=1$). To do so, following the steps mentioned in Algorithm \ref{Alg:1}, we design
	\begin{align*}
		P &= \begin{bmatrix}
			0.0099 & -0.0092 & 0.0036 & 0.0019 & -0.0054\\ -0.0092 & 0.0110 & -0.0067 & -0.0020 & 0.0046\\ 0.0036 & -0.0067 & 0.0099 & 0.0020 & -0.0018\\ 0.0019 & -0.0020 & 0.0020 & 0.0023 & -0.0018\\ -0.0054 & 0.0046 & -0.0018 & -0.0018 & 0.0058
		\end{bmatrix}\!\!,\\
		\Theta &= \begin{bmatrix}
			0.2490 & 0.2490 & 0.0804 & -0.1039 & 0.0255
		\end{bmatrix}^{\! \top}\!\!\!\!, ~\hat{A}=-1, ~\hat{B}=1,\\
		\Xi &= \begin{bmatrix}
			2.2615 & -2.0573
		\end{bmatrix}^{\! \top}\!\!\!\!, ~\Psi = \begin{bmatrix}
			0.0045 & -0.0037
		\end{bmatrix}^{\! \top}\!\!\!\!, ~\hat{\kappa} = 3, ~ \varepsilon=1, \\ \alpha&=5.1859\times 10^{-4} , ~\rho=7.5529\times10^{-5}.
	\end{align*}
	
	\begin{figure}[t!]
		\centering
		\subfloat[\centering 30 trajectories with arbitrary initial conditions \label{ex1: traj}]{{\includegraphics[width=.4\linewidth]{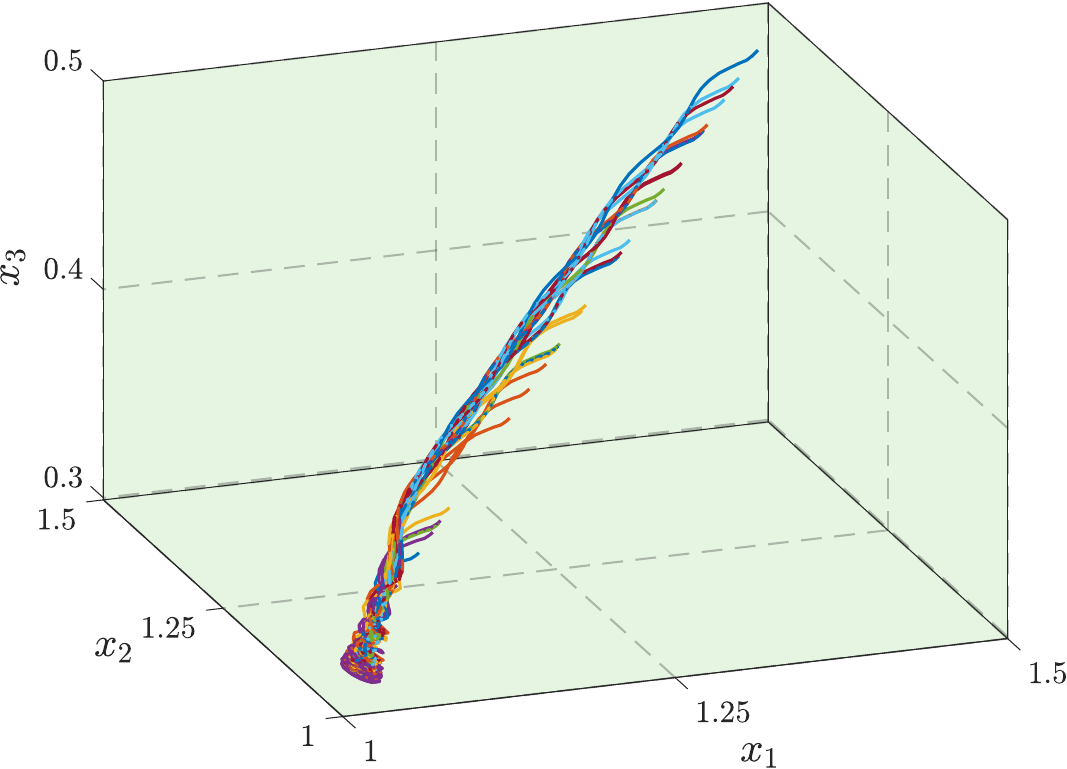} }}%
		\qquad\qquad
		\subfloat[\centering Errors between unknown system and its ROM \label{ex1: error}]{{\includegraphics[width=.4\linewidth]{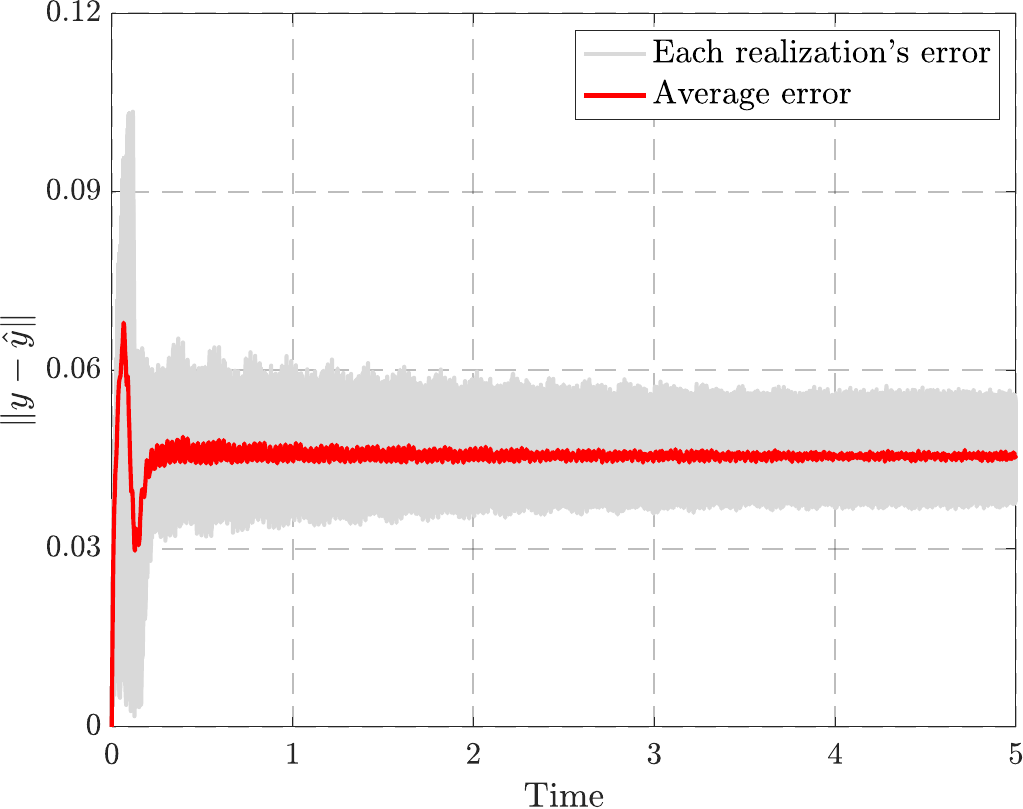} }}%
		\caption{(a) Closed-loop trajectories of the original system, starting from $30$ distinct initial conditions, together with (b) the closeness between trajectories of the \emph{motor system} and those of its corresponding ROM, as well as the mean of errors. As illustrated, all trajectories of the original unknown motor remain within the safe region under the designed ROM controller. \label{fig: example 1}}%
	\end{figure}
	
	We now synthesize a controller $\hat u \in [-10, 10]$ for the data-driven ROM using \texttt{SCOTS}~\cite{rungger2016scots} and refine it back to the original system using the interface map $u = u_{\hat u}(x, \hat x, \hat u)$ such that the first three states of the motor are forced to stay in a predetermined safe set $X = \left[1, \, 1.5\right] \times \left[1, \, 1.5\right]\times \left[0.3, \, 0.7\right]$. By considering $x_0 = \Theta \hat x_0$ and according to \eqref{new error}, the closeness between output trajectories of motor control system and its data-driven ROM is computed as $0.7282$. As can be seen in Figure~\ref{fig: example 1}, starting from $30$ different initial conditions, the original system's states are enforced to stay inside the safe set utilizing the controller solely synthesized for the ROM.
	
	\subsection{Spacecraft docking maneuvers} \label{subsec4.2}
	To showcase the efficacy of our data-driven framework for more complex logic properties like \emph{reach-while-avoid}, we utilize our results to solve a motion planning problem for spacecraft docking maneuvers~\cite{danielson2016path}. The system's matrices, which are assumed to be unknown to us, are as follows:
	\begin{align*}
		A = & \begin{bmatrix}
			0    &     0  &  1    &     0\\
			0    &     0    &     0  &  1\\
			-3.1623  &  0.0017 &  -4.0404  &  0.0022\\
			-0.0017 &  -3.1623 &  -0.0022 &  -4.0404
		\end{bmatrix}\!\!,\\
		B = & \begin{bmatrix}
			0  &   0  &   1  &   0\\
			0  &   0   &  0   &  1
		\end{bmatrix}^{\!\top}\!\!\!\!\!.
	\end{align*}
	
			\begin{figure}[t!]
		\centering
		\subfloat[\centering Trajectories of spacecraft and its ROM \label{ex2: comp}]{{\includegraphics[width=0.4\linewidth]{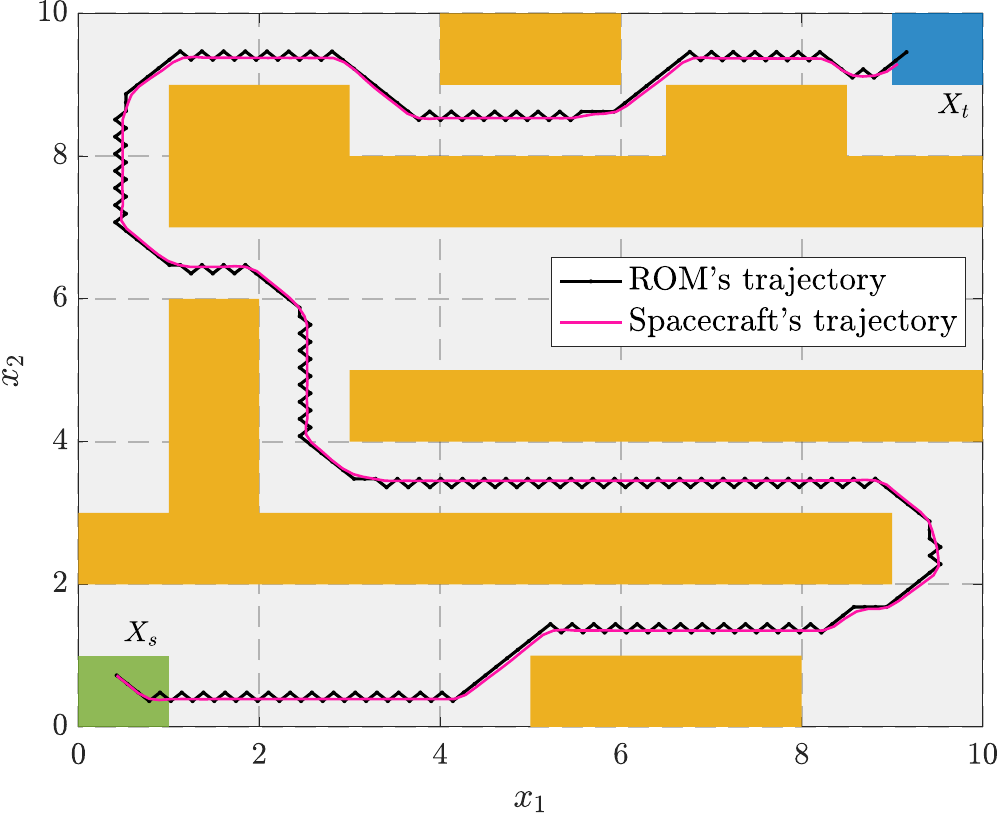} }}% 
		\qquad\qquad
		\subfloat[\centering $50$ arbitrary trajectories of spacecraft \label{ex2: multi}]{{\includegraphics[width=0.4\linewidth]{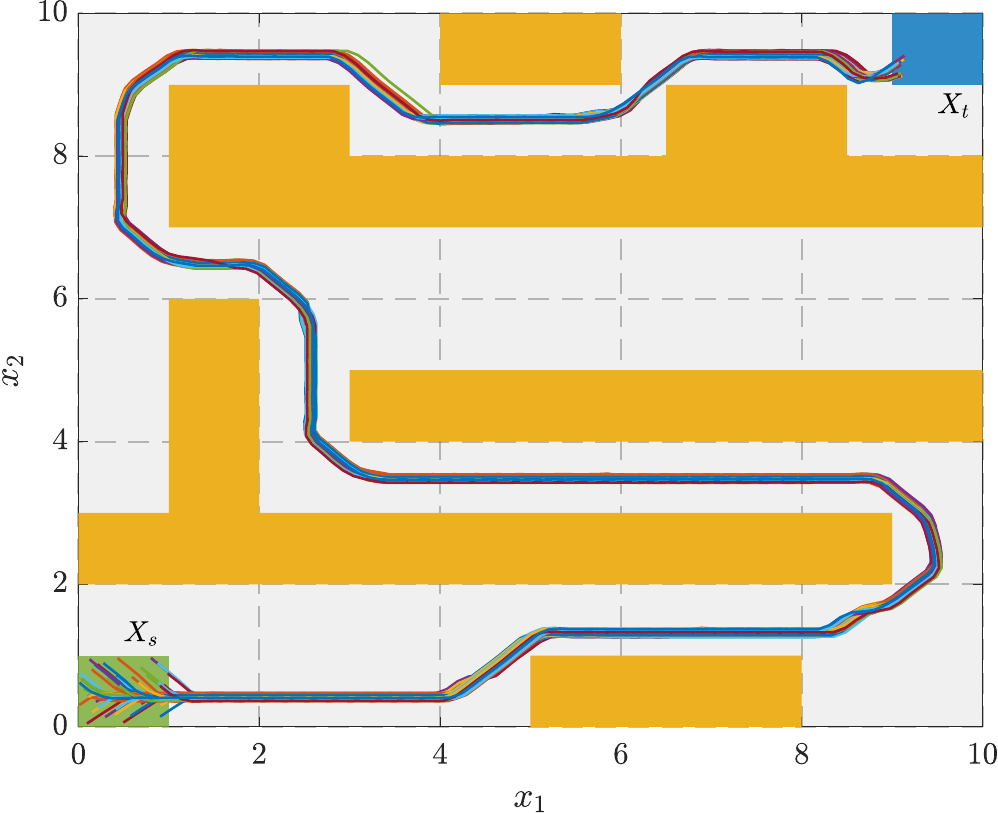} }}%
		\caption{Trajectories should start from the initial set $X_s$ \legendsquare{START!80} and reach the target set $X_t$ \legendsquare{TARGET!80} while avoiding debris \legendsquare{OBSTACLES!80} in space. As illustrated in (a), trajectories of the spacecraft and its ROM are very close to one another, satisfying the spacecraft's reach-while-avoid property. As shown in (b), the specification is properly fulfilled for $50$ arbitrary trajectories of the spacecraft under the designed ROM controller.\label{fig: example 2}}%
	\end{figure}
	
	For the reach-while-avoid property, the spacecraft's position should start from a specific initial set $X_s = [0,\, 1]\times [0,\, 1]$ and reach a predetermined target set $X_t = [9,\, 10]\times [9,\, 10]$ while avoiding debris in space. Our primary objective is to reduce the system's dimensions from 4 to 2, followed by designing a controller for the ROM that meets the reach-while-avoid specification for the spacecraft. After following the steps of Algorithm~\ref{Alg:1}, we have
	\begin{align*}
		P & = 10^{-6} \times \begin{bmatrix}
			0.4622 &  -0.0556  &  0.0762  & -0.0158\\
			-0.0556  &  0.4739  &  0.0003  &  0.0945\\
			0.0762  &  0.0003 &   0.0228&   -0.0017\\
			-0.0158  &  0.0945 &   -0.0017 &   0.0296
		\end{bmatrix}\!\!, \\ \Theta & = \begin{bmatrix}
			1 & 0\\
			0 & 1\\
			-0.0001 & 0\\
			0 & -0.0001
		\end{bmatrix}\!\!, ~\hat{A} =  \begin{bmatrix}
			-10^{-4} & 0\\ 0 & -10^{-4}
		\end{bmatrix}\!\!, ~\hat{B} = \begin{bmatrix}
			0.1 & 0\\ 0 & 0.1
		\end{bmatrix}\!\!,\\~
		\Xi & = \begin{bmatrix}
			3.1619 & -0.0017\\ 0.0017 & 3.1619
		\end{bmatrix}\!\!, ~\Psi = \begin{bmatrix}
			0.3313 & 0.0256\\ -0.0342 & 0.3203
		\end{bmatrix}\!\!,  ~\hat{\kappa}=5, ~\varepsilon = 1,\\
		\alpha &= 9.0420\times 10^{-9}, ~\rho=2.1718 \times 10^{-9}.
	\end{align*}
	Using the obtained parameters and by setting $x_0 = \Theta \hat{x}_0$, the closeness between the output trajectories of the system and its ROM is calculated as $0.3603$, according to \eqref{new error}.
	
	Leveraging the obtained ROM, we use \texttt{SCOTS} to design the desired controllers $\hat u_1, \hat u_2 \in [-6,6]$ for the ROM, followed by refining it back to the original system. The simulation results are provided in Figure~\ref{fig: example 2}, showcasing the practicality of our framework for such a complex motion planning task.
	
	\subsection{Human blood glucose metabolism}
	Controlling human blood glucose metabolism can be considered another case study of our work~\cite{sorensen1985physiologic}. The dynamic of this safety-critical application has seven states, including brain, heart, gut, lungs, kidney, and periphery (with two states in this dynamic), making it exceedingly difficult to design a formal controller that satisfies complex specifications. 
	The system's matrices, which are assumed to be unknown to us, are as follows:
	\begin{align*}
		A \!=\! \left[\!\!\!\begin{array}{ccccccc}
			-1.73 \!\!&\!\!   1.73   \!\!&\!\!        0     \!\!&\!\!      0  \!\!&\!\! 0  \!\!&\!\!  0     \!\!&\!\!     0\\
			0.45  \!\!&\!\!   -3.15    \!\!&\!\!       0   \!\!&\!\!   0.90  \!\!&\!\!  0.72 \!\!&\!\!  1.06     \!\!&\!\!      0\\
			0   \!\!&\!\!   0.76  \!\!&\!\!    -0.76    \!\!&\!\!       0  \!\!&\!\!  0  \!\!&\!\!  0  \!\!&\!\!  0\\
			0   \!\!&\!\!   0.08  \!\!&\!\!    0.32  \!\!&\!\!  -0.76  \!\!&\!\!  0 \!\!&\!\!  0  \!\!&\!\!  0\\
			0  \!\!&\!\!    1.41     \!\!&\!\!      0     \!\!&\!\!      0  \!\!&\!\!  1.19  \!\!&\!\!  0    \!\!&\!\!    0\\
			0  \!\!&\!\!    1.41    \!\!&\!\!       0     \!\!&\!\!      0  \!\!&\!\!  0  \!\!&\!\! -1.87  \!\!&\!\!    0.45\\
			0     \!\!&\!\!      0   \!\!&\!\!        0    \!\!&\!\!       0  \!\!&\!\! 0  \!\!&\!\!  0.05 \!\!&\!\! -0.46
		\end{array}\!\!\!\right]\!\!,   \quad B \!=\! \begin{bmatrix}
			0\\
			1\\
			0\\
			0\\
			0\\
			0\\
			0
		\end{bmatrix}\!\!.
	\end{align*}
	
	Our framework can remarkably alleviate this complexity by reducing the system's dimensions and designing the desired controller for the ROM. The main goal is to reduce the dimensions of the original system to 1, followed by synthesizing a controller, standing for the insulin injected into the body by an external insulin pump, that steers the insulin concentration in the human body (\emph{i.e.,} the summation of all states) to be between $[30,40]$.
	According to the steps of Algorithm \ref{Alg:1}, we design
	\begin{align*}
		P & = \left[\!\!\!\begin{array}{ccccccc}
			0.0036 & -0.0002 & -0.0038 & 0.0004 & -0.0020 & -0.0011  &  0.0662\\
			-0.0002 & 0.0006 & 0.0016 & 0.0003 & 0.0024 & 	-0.0016 &  -0.0453\\
			-0.0038 & 0.0016 & 0.0240 & 0.0053 & 0.0168 & -0.0173 &  -0.5736\\
			0.0004 & 0.0003 & 0.0053 & 0.0031 & 0.0031 &  -0.0051  & -0.1336\\
			-0.0020 & 0.0024 & 0.0168 & 0.0031 & 0.0201 &  -0.0153 &  -0.4722\\
			-0.0011 & -0.0016 & -0.0173 & -0.0051 & -0.0153 &  0.0189 &   0.4709\\
			0.0662 & -0.0453 & -0.5736 & -0.1336 & -0.4722 & 0.4709 &  14.6800
		\end{array}\right]\!\!,\\
		\Theta & = \left[\!\!\!\begin{array}{c}
			0.1553\\
			0.1552\\
			0.1554\\
			0.0832\\
			-0.1831\\
			0.1208\\
			0.0131\\
		\end{array}\!\!\!\right]\!\!, \qquad\begin{cases}
			\hat{A}=-0.001,\\
			\hat{B} = 1,\\
			\Xi = 0.3479,\\
			\Psi = -1.5911,\\
			\hat{\kappa}=1.5,\\
			\varepsilon = 1,\\
			\alpha = 1.4161 \times 10^{-4}.
		\end{cases}
	\end{align*}
	
	Now, similar to the first case study, one can readily synthesize a controller for the ROM such that $\hat x$ is forced to be between 60 and 80 (as the summation of $\Theta$ is approximately 0.5). In this way, the insulin concentration in the human body is steered to be between $[30,40]$.
	
	\subsection{Comparison with recent literature: Cart system}
	This case study is a cart system with a double-pendulum controller. This system has six states and one control input, with the system's matrices as
	\begin{align*}
		A &= \begin{bmatrix}
			0  &  1     &    0     &    0     &    0     &    0\\
			-1  & -1 &  19.6 &   1   &      0   &      0\\
			0     &    0    &     0  &  1    &     0     &    0\\
			1  &  1 &  -39.2 &  -2  &  9.8  &  1\\
			0     &    0     &    0     &    0     &    0  &  1\\
			0    &     0  & 19.6  &  1  & -19.6 &  -2
		\end{bmatrix}\!\!, \quad B = \begin{bmatrix}
			0\\
			1\\
			0\\
			-1\\
			0\\
			0
		\end{bmatrix}\!\!.
	\end{align*}
	Remark that we assume these matrices are unknown to us. Following the steps mentioned in Algorithm \ref{Alg:1}, we design
	\begin{align*}
		P &= \begin{bmatrix}
			7.2436  &  3.6926  &  3.2721  &  2.9688  &  2.3036  &  1.9298\\
			3.6926 &   2.5014 &   1.1149  &   1.9590 &   0.8084 &   1.2664\\
			3.2721 &   1.1149 &   2.6025  &  0.9937  &  1.7212   & 0.6680\\
			2.9688  &  1.9590 &   0.9937  &  1.5670   & 0.7002  &  1.0204\\
			2.3036 &   0.8084  &  1.7212  &  0.7002  &  1.4650   & 0.4834\\
			1.9298&    1.2664  &  0.6680 &   1.0204  &  0.4834  &  0.6772
		\end{bmatrix}\!\!,\\ \Theta & = \begin{bmatrix}
			0.0501\\
			-0.0001\\
			0\\
			0\\
			0\\
			0\\
		\end{bmatrix}\!\!, \qquad\begin{cases}
			\hat{A} = -0.0020,\\
			\hat{B} = 1,\\
			\Xi = 0.05,\\
			\Psi = 0.2409,\\
			\hat{\kappa} = 2,\\
			\varepsilon = 1,\\
			\alpha = 0.0047.
		\end{cases}
	\end{align*}
	
	Our primary aim is to apply our MOR findings to this system to compare our proposed data-driven method with the recent literature~\cite{ionescu2015two,burohman2023data}. For comparison, our framework reduces the system's dimension from 6 to 1, whereas the ROM in~\cite{ionescu2015two} reduces it to 2 (though it is not data-driven), and the data-driven ROM in~\cite{burohman2023data} reduces it to 3. This demonstrates the superior effectiveness of our framework. Additionally, the guarantee in~\cite{burohman2023data} is solely limited to the stability property, whereas our work offers a specification-free framework (cf. Remark \ref{Remark new}). This indicates that once the ROM is obtained from data, one can enforce complex specifications across the original system by satisfying them over its data-driven ROM.
	
	\subsection{High-dimensional system}
	To demonstrate the practicality of our method for high-dimensional systems, we consider a dynamical system with 25 states and two control inputs, \textit{i.e.,} $A\in\mathbb{R}^{25\times 25}$ and $B\in\mathbb{R}^{25\times 2}$. Directly synthesizing a controller for this size system using formal-method tools is infeasible. However, utilizing our framework, we reduce the system's dimensions to $2$, synthesize a controller for the ROM, and then refine it back to the original system to solve a \emph{reach-while-avoid} problem.
	More precisely, we aim for the trajectories of $x_1$ and $x_2$ to commence from the initial set $X_s = [4,\, 5]\times [0,\, 1]$ and be steered to the target set $X_t = [0,\, 1]\times [3.5,\, 6]$ without a collision transpiring with the obstacles. Figure \ref{fig:rwa25} depicts the satisfaction of the reach-while-avoid property for this high-dimensional system, while different errors of different runs are shown in Figure \ref{fig:error25}. Finally, the behavior of all the states of the system is illustrated in Figure \ref{fig:25states}.
	
	\begin{figure}[t!]
		\centering
		\includegraphics[width=0.4\linewidth]{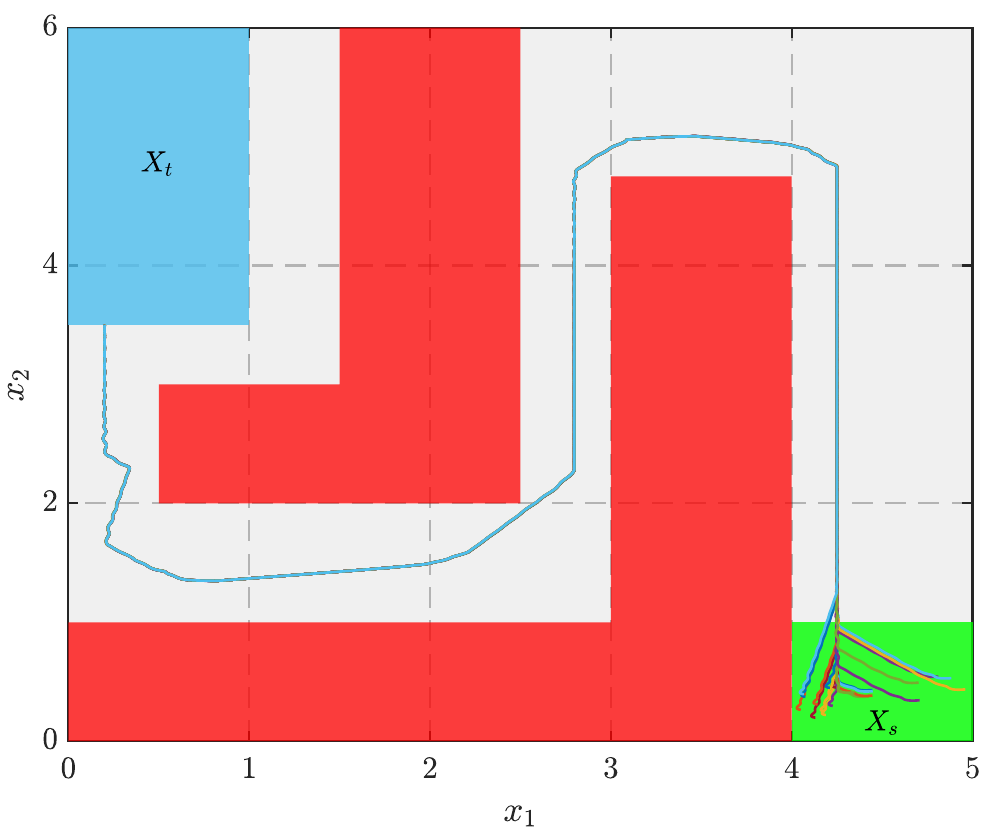}
		\caption{As the reach-while-avoid specification, trajectories of $x_1$ and $x_2$ should start from the initial set $X_s$ \legendsquare{START1!80} and reach the target set $X_t$ \legendsquare{TARGET1!80} without colliding with the obstacles \legendsquare{OBSTACLES1!75}. As seen, the specification is correctly fulfilled for $20$ arbitrary trajectories of the original system under the designed ROM controller, offering the practicality of our framework for high-dimensional systems.}
		\label{fig:rwa25}
	\end{figure}
	
	\begin{figure}[t!]
		\centering
		\includegraphics[width=0.4\linewidth]{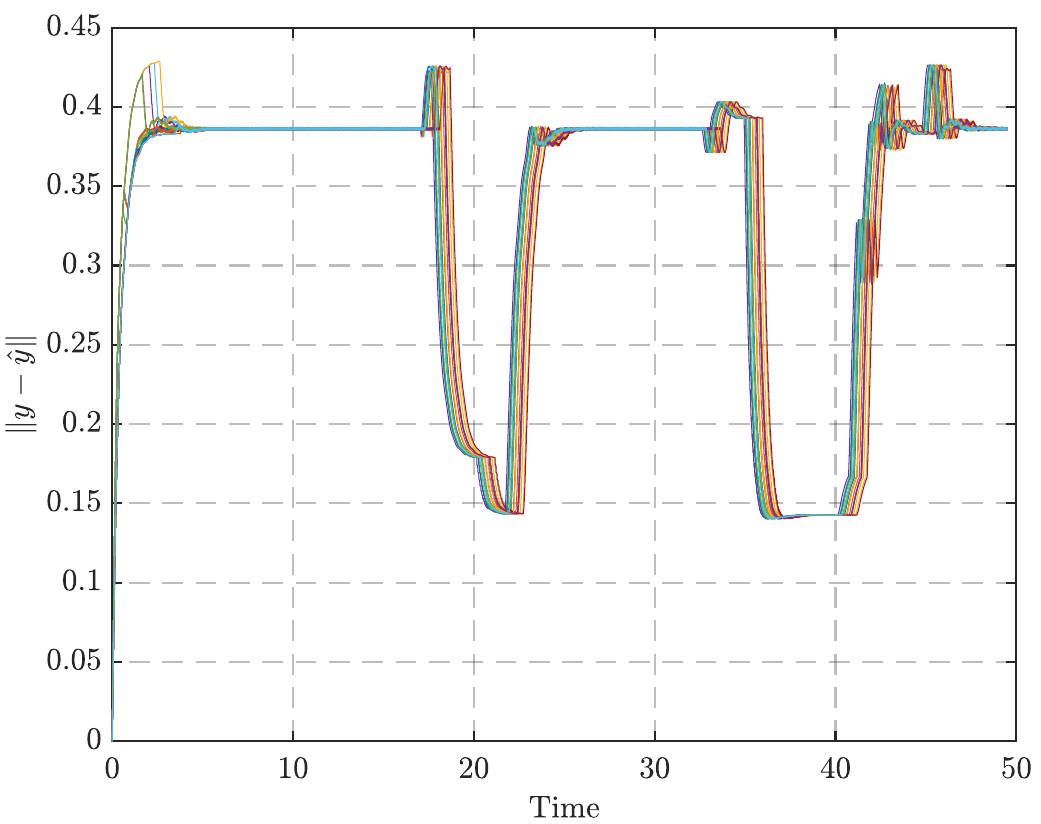}
		\caption{The quantified error bound for this example is $0.8609$ according to Theorem \ref{thm-J19}. As seen, all plotted errors are below the quantified upper bound $0.8609$, demonstrating the formality of our results.}
		\label{fig:error25}
	\end{figure}
	
	\begin{figure}[t!]
		\centering
		\includegraphics[width=0.4\linewidth]{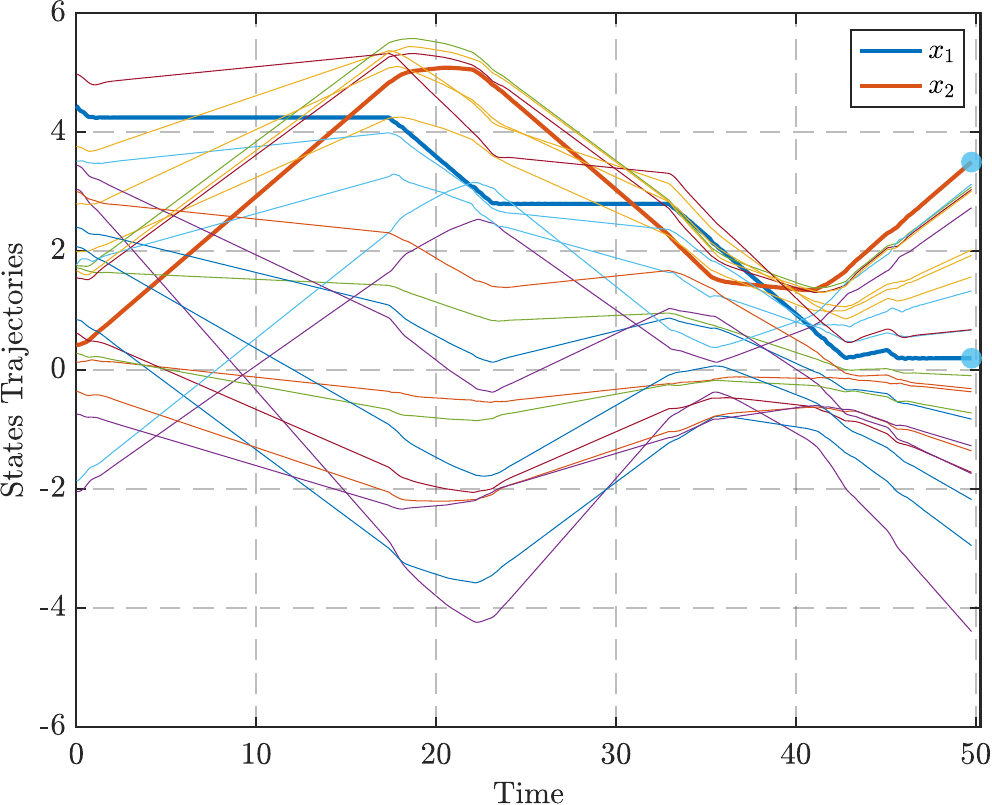}
		\caption{The behavior of all states while the reach-while-avoid is solved for the first and second states. As illustrated, there is no anomaly in the behavior of states, and the desired property is clearly satisfied, \textit{i.e.,} $x_1$ and $x_2$ reached the target
			{\scriptsize \legendcircle{TARGET1!80}}.
		}
		\label{fig:25states}
	\end{figure}
	
	\section{Conclusion}\label{conc}
	In this work, we introduced a certified data-driven scheme to construct reduced-order models of dynamical systems with unknown mathematical models. Our methodology leveraged data to establish similarity relations between output trajectories of unknown systems and their data-driven ROMs using simulation functions, which formally quantify their closeness. To achieve this, under a readily fulfillable rank condition using data, only two input-state trajectories from unknown systems were collected to construct both ROMs and SFs while providing correctness guarantees. By leveraging a series of benchmarks, we demonstrated that the proposed ROMs derived from data can be used for controller synthesis, ensuring high-level logical properties over unknown dynamical models. Constructing ROMs from data for a \emph{class of nonlinear systems}, as in \eqref{Nonlinear}, is under investigation as future work.

\bibliographystyle{alpha}
\bibliography{biblio}

\end{document}